\documentclass[reqno]{amsart}
\usepackage{graphicx}
\usepackage{enumitem}
\usepackage[numbers,sort&compress]{natbib}
\usepackage{hyperref}
\usepackage[foot]{amsaddr}
\usepackage[final,color,notref,notcite]{showkeys}
\usepackage{caption}

\usepackage[charter]{mathdesign}
\usepackage[scaled=.96,lf]{XCharter}
\usepackage[scr=rsfso, scrscaled=.98]{mathalfa}

\definecolor{labelkey}{rgb}{0,.56,.7}
\presetkeys{todonotes}{color=red}{}

\DeclareMathAlphabet{\pazocal}{OMS}{zplm}{m}{n}   

\newcommand{\Lcal}{\pazocal{L}}

\newtheorem{thm}{Theorem}
\newtheorem{prop}[thm]{Proposition}
\newtheorem{lem}[thm]{Lemma}

\theoremstyle{remark}
\newtheorem*{rem}{\bf Remark}
\newtheorem{defi}{\bf Definition}

\newtheorem{exaa}{\bf Example}

\newenvironment{exa}
  {\pushQED{\qed}\exaa}
  {\popQED\endexaa}

\newcommand*{\at}{@}

\newcommand{\nn}{\nonumber}

\def\dg{\dagger}
\def\df{\overset{\mathrm{df}}{=}}
\def\ba{\boldsymbol{\a}}
\newcommand{\ket}[1]{\mathop{|#1\rangle}\nolimits}
\newcommand{\bra}[1]{\mathop{\left<#1\,\right|}\nolimits}

\newcommand{\kbr}[2]{| #1\rangle\!\langle #2 |}

\def\ran{\rangle}
\def\lan{\langle}

\DeclareMathOperator*{\haf}{Haf}

\newcommand{\diff}[2]{\mathrm{d}^#1 #2}             
\newcommand{\dif}[1]{\mathrm{\,d} #1}             
\newcommand{\difff}[1]{\mathrm{\,d}^3 #1}             

\def\a{\alpha}

\def\g{\gamma}
\def\G{\Gamma}
\def\d{\delta}

\def\om{\omega}
\def\Om{\Omega}
\def\s{\sigma}

\def\la{\lambda}

\def\bbZ{\mathbb{Z}}

\def\bbE{\mathbb{E}}

\DeclareSymbolFont{Eulerscripteusm10}{U}{eus}{m}{n}
\SetSymbolFont{Eulerscripteusm10}{bold}{U}{eus}{b}{n}
\DeclareMathSymbol{\euA}{\mathord}{Eulerscripteusm10}{"41}
\DeclareMathSymbol{\euI}{\mathord}{Eulerscripteusm10}{"4A}

\begin{document}

\title[A novel approach to perturbative calculations for a large class of interacting boson theories]{{A novel approach to perturbative calculations for a large class of interacting boson theories}}

\begin{abstract}
We present a method of calculating the interacting $S$-matrix to an arbitrary perturbative order for a large class of boson interaction Lagrangians. The method takes advantage of a previously unexplored link between the $n$-point Green's function and a certain system of linear Diophantine equations. By finding all nonnegative solutions of the system, the task of perturbatively expanding an interacting $S$-matrix becomes elementary for any number of interacting fields, to an arbitrary perturbative order (irrespective of whether it makes physical sense) and for a large class of scalar boson theories. The method does not rely on the position-based Feynman diagrams and promises to be extended to many perturbative models typically studied in quantum field theory. Aside from interaction field calculations we showcase our approach by expanding a pair of Unruh-DeWitt detectors  coupled to Minkowski vacuum to an arbitrary perturbative order in the coupling constant. We also link our result to Hafnian as introduced by Caianiello and present a method to list all $(2n-1)!!$ perfect matchings of a complete graph on $2n$ vertices.
\end{abstract}


\author{Kamil Br\'adler}

\email{kbradler\at uottawa.ca}

\address{Xanadu, 372 Richmond St W, Toronto, M5V 2L7, Canada}
\address{Department of Mathematics and Statistics, University of Ottawa, Ottawa, Canada}

\maketitle

\thispagestyle{empty}
\allowdisplaybreaks

\section{Introduction}\label{sec:intro}

The calculation of an interacting $S$-matrix is one of the first problems encountered in quantum field theory (QFT)~\cite{schwartz2014quantum} . The majority of physical models must be calculated perturbatively in the coupling constant and there exists a whole calculational industry with one sole purpose: to make the increasingly tedious calculations manageable for a large variety of interaction Lagrangians~\cite{hahn2001generating,cullen2014gosam,belanger2006automatic,shtabovenko2016new,alwall2014automated,tentyukov2000feynman,christensen2009feynrules}.
Selected theories (such as the $\phi^4$ model) were studied even in more detail providing many-loop expansions in terms of the Feynman diagrams including closed expressions for their multiplicity factors~\cite{kleinert2000recursive,kajantie2002simple,palmer2002general}. These works fall into a broader effort of the Feynman diagram enumeration~\cite{bender1976statistical,cvitanovic1978number,hue2012general,hurst1952enumeration,mestre2006generating,brouder2009quantum}, including generalizations to various field theories~\cite{kleinert2002recursivee,pelster2003functional,bachmann2000recursive,pelster2004recursive,pelster2003many,pelster2002functional}.

The purpose of this paper is to develop a different method to construct perturbative contributions which is based on (what seems to be) an unexplored link between Wick's theorem~\cite{wick1950evaluation} (or its statistics equivalent due to Isserlis~\cite{isserlis1918formula}) and a certain system of linear Diophantine equations. Based on this insight, our approach is readily applicable to a large class of boson scalar theories and have a great chance of being generalized to a much larger class of theories typically considered in QFT. It is not our goal to compete with the aforementioned multipurpose packages to perform perturbative calculations. Rather, we hope that our method will offer a new conceptual insight into perturbative calculations where a speedy calculation of scattering amplitudes and multiplicities would be a nice bonus. Indeed, from the practical point of view, the method provides an extremely streamlined and versatile way of calculating the $S$-matrix contributions including multiplicities to a high perturbative order, for any number of interacting fields and for a large class of real scalar bosonic theories without the need to summon Feynman's diagrams   (all position-based Feynman diagrams are automatically obtained  in the Diophantine approach). A strategy to go beyond scalar boson models is outlined as well.

The effortless nature of the presented method is quite surprising due to the fact that solving Diophantine equations has a well-earned reputation of being hard. In more detail, it turns out that the number of ways to simplify an $n$-point Green's function is given by all nonnegative solutions of a linear Diophantine system. This, on the other hand, is a problem equivalent to counting the number of interior points of convex polyhedra -- a major topic in algebraic geometry~\cite{beck2007computing,stanley1997enumerative}.

There are circumstances where the increase of perturbative components does not make sense from the physical point of view. It is well known  that some theories are (super)renormalizable and others are not (depending of the spacetime dimension $d$)~\cite{schwartz2014quantum}. Additionally, infinities of a different type creep even into those theories which are renormalizable because their perturbative expansion is asymptotic~\cite{dyson1952divergence,Lipatov1976ny,zinn1981perturbation}. A typical example is the quantum electrodynamics Lagrangian~\cite{schwartz2014quantum} whose perturbative contributions start to increase at the order of the inverse of the fine-structure constant. Nonetheless, we feel that the generality of the presented method and its potential to go beyond boson theories together with so far unnoticed connection (to the best author's knowledge) of perturbative calculations to algebraic geometry/number theory  makes the method interesting and we return to the potentially fruitful link between QFT perturbative calculations and lattice polyhedra in the last section.

It may seem that the number of real scalar boson theories in QFT that are worth of exploring to any order and for any number of interacting fields is limited. Even if it was true, there exists a number of physical processes where such expansions are desirable. One of them is a model  of a real scalar field linearly coupled to a two-level quantum system known as the Unruh-DeWitt (UDW) particle detector. It was conceived in~\cite{unruh1976notes}, improved~\cite{dewitt1979quantum} and studied in many physical situations~\cite{svaiter1992inertial,higuchi1993uniformly,ver2009entangling,reznik2005violating,schlicht2004considerations,lin2006accelerated,louko2006often,sriramkumar1996finite,barbado2012unruh,cliche2010information,franson2008generation,bradler2016absolutely}. The starting point for perturbative calculations is  the Hamiltonian formalism (equivalent to the Lagrangian approach in the absence of a derivative coupling). But one quickly notices a qualitative difference compared to the perturbative expansions in QFT since the problem does not seem to suffer from the convergence problems~\cite{dyson1952divergence,Lipatov1976ny}. Contrary to the typical situation in QFT, it is shown that the number of perturbative terms grows polynomially with the perturbative order. It is not clear whether an actual non-perturbative solution exists (we leave it as an interesting open problem) but the main result of our calculation based on a Diophantine system of linear equations is the next best thing: an efficiently calculable expansion to an arbitrary order in the coupling constant.

Finally, we explore the connection to the so-called Hafnian introduced in~\cite{caianiello1953quantum} in the context of Dyson's series expansion. This insight enables to apply our Diophantine algorithm to list all $(2n-1)!!$ perfect matchings of a complete graph on $2n$ vertices. Hafnians have recently attracted a lot of attention and their applications go well beyond perturbative QFT~\cite{barvinok2016approximating,rudelson2016hafnians,schultz1992topological,krenn2017quantum}.

The paper is organized as follows. After a brief recollection of the scalar boson theory in Section~\ref{sec:bosonsInteracting} through its Lagrangian formulation  we present the main findings of this paper in Section~\ref{sec:DiophantineSols}. This is accompanied by several examples in Section~\ref{sec:apps} of the perturbative contributions  of the $\phi^4$ and $\phi^3$ theory (for an easy comparison with the published results) and their multiplicities  and one major application which is the perturbative calculation of a pair of  UDW detectors to an arbitrary order in the coupling constant. In Section~\ref{sec:beyondscalars} we discuss the necessary steps in order to generalize the current formalism to more complicated perturbative models in QFT and Section~\ref{sec:discussions} concludes with several open problems.

\section{Boson scalar theories}\label{sec:bosonsInteracting}

An important object of study in interacting QFT is the S-matrix, which is proportional to  the $n$-point interacting Green's  function
\begin{equation}\label{eq:Smatrix}
  \bra{\Om}\mathsf{T}\{\boldsymbol{\phi}(x_1)\dots\boldsymbol{\phi}(x_k)\}\ket{\Om},
\end{equation}
where  $\ket{\Om}$ is the interacting vacuum for the given theory, $\boldsymbol{\phi}(x_i)$ are the interacting fields and $\mathsf{T}$ stands for the time-ordering operator. In the Lagrangian formulation of the theory the way of calculating  Green's  function~\eqref{eq:Smatrix} is through the formula~\cite{schwartz2014quantum}
\begin{equation}\label{eq:GreenInteracting}
  \bra{\Om}\mathsf{T}\{\boldsymbol{\phi}(x_1)\dots\boldsymbol{\phi}(x_k)\}\ket{\Om}
  ={\bra{0_M}\mathsf{T}\{\phi(x_1)\dots\phi(x_k)\exp{[i\int\diff{d}{z}\Lcal_{\mathrm{int}}]\}}\ket{0_M}\over\bra{0_M}\mathsf{T}
  \{\exp{[i\int\diff{d}{z}\Lcal_{\mathrm{int}}]}\}\ket{0_M}},
\end{equation}
where the interacting part of the complete Lagrangian $\Lcal=\Lcal_0+\Lcal_{\mathrm{int}}$ appears, $\phi(x_i)$ are free fields and $\ket{0_M}$ is a free ground state in $d$-dimensional Minkowski spacetime (Minkowski vacuum). The type of theories we will study in this work are the $N$-mode boson scalar theories whose free Lagrangian reads
\begin{equation}\label{eq:freeLagr}
  \Lcal_0=\sum_{k=1}^N{1\over2}\partial^\nu\phi_k\partial_\nu\phi_k-{1\over2}\mu_k^2\phi_k^2
\end{equation}
and the interaction part can take many different forms. Instead of trying to write the most general $\Lcal_{\mathrm{int}}$, we will list several types of interaction (omitting the multiplicative terms making the Lagrangian densities of dimension~$m^d$):
\begin{align}\label{eq:LintListed}
  \Lcal_{\mathrm{int}} & = -\sum_{n=3} g_n{\phi^n\over n!},\\
  \Lcal_{\mathrm{int}} & \propto \prod_{l=1}^Ng\phi_l^{n_l},\\
  \Lcal_{\mathrm{int}} & = -{g\over4}\bigg(\sum_{l=1}^{N}\phi_i\phi_i\bigg)^2.
\end{align}
In the first line we set $N=1$, in the second item we skipped the numerical prefactors and in the third line we have $\mu_k=\mu$ with $(\phi_1,\dots,\phi_N)$ forming a $N$-component boson field exhibiting the $O(N)$ symmetry (the so-called $O(N)$ sigma model). The fields in~\eqref{eq:LintListed} are real but the current analysis is equally applicable to complex fields  with only minor modifications (and for suitable interaction Lagrangians -- we discuss this and other possible generalizations in Section~\ref{sec:beyondscalars}). This is due to the fact that $\phi$ and $\phi^\dg$ must be considered independent. After all, the $O(2)$ sigma model can be rewritten as a one-mode complex theory whose interaction term reads $\Lcal_{\mathrm{int}}=-g/4(\phi^\dagger\phi)^2$.

By looking at the RHS of~\eqref{eq:GreenInteracting}, the solution is given by expanding the numerator around the coupling parameter. The denominator (the free $S$-matrix) is known to factor out and this removes the disconnected contributions of the scattering amplitudes in the numerator. Here comes our main contribution. We present an efficient way of calculating the perturbative expansion of the numerator of~\eqref{eq:GreenInteracting} for arbitrary number $k$ of interacting fields $(\phi(x_i))_{i=1}^k$, for an extensive class of interaction Lagrangians such as those in Eqs.~\eqref{eq:LintListed} and to any perturbative order in the coupling constant. A generic expansion element in the numerator reads
\begin{equation}\label{eq:SmatrixOrderm}
  S^{(m)}={i^m\over m!}
  \int\diff{d}{\boldsymbol{z}}\bra{0_M}\mathsf{T}\big\{\phi(x_1)\dots\phi(x_k)\Lcal_{\mathrm{int}}(z_1)\dots\Lcal_{\mathrm{int}}(z_m)\big\}\ket{0_M}
\end{equation}
and the expectation value (the $k+m$-point Green's function) is the focus of this work. Here we make the perturbative calculations extremely straightforward for \emph{any} $m$ and $k$  without relying on Feynman rules or constructing Feynman position-space diagrams. We show how to efficiently factorize the $k+m$-point Green's function for large $k,m$ into a sum of products of two-point Green's functions (Feynman propagators) with no effort whatsoever. This is typically the most tedious step when dealing with perturbation techniques.

Lagrangians may also contain derivative interactions $\Lcal_{\mathrm{int}}\equiv\Lcal_{\mathrm{int}}(\phi_i,\partial_\nu{\phi_i})$. This is more or less a technical issue despite  the fact that the identity
\begin{equation*}
  \bra{0}\mathsf{T}\{\partial_\mu\phi(x)\partial_\nu\phi(y)\}\ket{0}
  =\partial_\mu\partial_\nu\bra{0_M}\mathsf{T}\{\phi(x)\phi(y)\}\ket{0_M}-g_{\mu0}\d(x'-y')\bra{0_M}\mathsf{T}\{\phi(x)\partial_\nu\phi(y)\}\ket{0_M}
\end{equation*}
complicates the extraction of the derivatives out of the propagator. It was shown in~\cite{rohrlich1950quantum} that the delta function contributions cancel out in the perturbative expansion of the whole $S$-matrix. Hence, in principle, we could use the `essentially' correct identity
$$
\bra{0}\mathsf{T}\{\partial_\mu\phi(x)\partial_\nu\phi(y)\}\ket{0}\ \mbox{`='}\ \partial_\mu\partial_\nu\bra{0_M}\mathsf{T}\{\phi(x)\phi(y)\}\ket{0_M}
$$
and perform the same analysis we will present here.

\section{$n$-point Green's functions solved through a system of linear Diophantine equations}\label{sec:DiophantineSols}

\begin{thm}[Isserlis'~\cite{isserlis1918formula}]\label{thm:Isserlis}
  Let $x_i$ be a Gaussian random variable satisfying $\bbE[\prod_{i=1}^{2m+1}x_i]=0$. Then
  \begin{equation}\label{eq:Isserlis}
    \bbE\Big[\prod_{i=1}^{2m}x_i\Big]=\sum_{r=1}^{(2m-1)!!}\prod_{\genfrac{}{}{0pt}{2}{j,k=1}{j<k}}^{m}\bbE_r[x_{j}x_{k}],
  \end{equation}
  where the sum goes over the products of bivariate expectation values $\bbE_r$.
\end{thm}
\begin{rem}[notational]
  In our case the Gaussian random variable will be a product of time-ordered free scalar fields $\phi(x_i)$ and $\phi(z_j)$ and the expectation value will be taken w.r.t. to Minkowski vacuum $\ket{0_M}$:
  \begin{equation}\label{eq:notation}
     \bbE\Big[\prod_{i=1}^{2m}x_i\Big]=\bra{0_M}\mathsf{T}\big\{\prod_{i=1}^{2m}\phi_{i}\big\}\ket{0_M}\equiv\lan\prod_{i=1}^{2m}i\ran_0,
  \end{equation}
  where  $\phi_{i}\equiv\phi(x_i)$ (or $\phi(z_i)$) and on the RHS is the minimalist notation we will be mostly using.
\end{rem}
\begin{rem}
  The theorem is also known as Wick's theorem~\cite{wick1950evaluation}. Wick's theorem transforms time-ordered operator expressions into a normal form~\cite{schwartz2014quantum} and Isserlis' result is recovered upon taking a (free) vacuum expectation value. As a matter of fact, Isserlis' theorem is more general since naturally there is no notion of time-ordering  in Eq.~(\ref{eq:Isserlis}) and so the theorem applies even for `ordinary' products of free scalar fields. Of course, we are actually calculating the time ordered version as it appears in the definition of the sought Green's function.
\end{rem}
For $2m$ different random variables $x_i$ (scalar fields $\phi_i$) in Eq.~(\ref{eq:notation}) there is nothing else to say but Isserlis' theorem can be refined if the number of \emph{different} variables $x_i$  in~\eqref{eq:Isserlis} is less than $2m$.
\begin{defi}[\cite{stanley1997enumerative}]\label{def:compositions}
  A $k$-composition of $n\in\bbZ_{>0}$ is $n$ written as an ordered sum of $k$ strictly positive integers.
\end{defi}
\begin{rem}
  Order matters unlike for integer partitions. There are $\binom{n-1}{k-1}$ $k$-compositions of $n$. For example, for $n=4$, there are three 2-compositions: $1+3,2+2$ and $3+1$.
\end{rem}
\begin{defi}\label{def:lexi}
  Given a finite set of positive integers $S=(1,2,\dots,f)$ we define the lexicographic ordering on the subset of two elements of $i,j\in S$ as $ij\leq i'j'$ iff $i\leq i'$ or $i=i'$ together with $j\leq j'$.
\end{defi}
\begin{thm}\label{thm:IsserlisRefined}
  Let $\ell_i\in\bbZ_{\geq0}$ such that $\sum_{i=1}^f\ell_i=2m$. Then,  Green's function
    \begin{equation}\label{eq:GreensFcns}
      \bra{0_M}\mathsf{T}\big\{\prod_{i=1}^{\ell_1}\phi_{1}\prod_{i=1}^{\ell_2}\phi_{2}\dots\prod_{i=1}^{\ell_f}\phi_{f}\big\}\ket{0_M}
      \equiv\lan1^{\ell_1}2^{\ell_2}\dots f^{\ell_f}\ran_0
      =\lan\underbrace{1\dots1}_{\ell_1}\underbrace{2\dots2}_{\ell_2}\dots\underbrace{f\dots f}_{\ell_f}\ran_0
    \end{equation}
  can be written as a product of two-point correlation functions
  \begin{equation}\label{eq:GreensProduct}
    \lan1^{\ell_1}2^{\ell_2}\dots f^{\ell_f}\ran_0
    =\sum_{\ba}\mu_{\ba}\prod_{\genfrac{}{}{0pt}{2}{i,j=1}{i<j}}^{f}\lan ij\ran^{\a_{ij}}_0,
  \end{equation}
  where $\mu_{\ba}\in\bbZ_{>0}$ is the product multiplicity, $\a_{ij}\in\bbZ_{\geq0}$ and  $\ba\df(\a_{ij})_{1\leq i\leq j\leq f}$. The number of products is a polynomial function of $\ell_i$  for a fixed number of fields $f$.
\end{thm}
\begin{defi}\label{def:coformations}
  A product $\prod_{\genfrac{}{}{0pt}{2}{i,j=1}{i<j}}^{f}\lan ij\ran^{\a_{ij}}_0$~is called a~\emph{conformation} and $\ba$ a \emph{conformation exponent}.
\end{defi}
\begin{rem}
  The conformation exponent $\ba$ will always be lexicographically ordered but it is not necessary for the proof of the above theorem.
\end{rem}
\begin{proof}[Proof of Theorem~\ref{thm:IsserlisRefined}]
  To solve Eq.~(\ref{eq:GreensProduct}) means to find all  nonnegative solutions $\a_{ij}$ of the system
    \begin{equation}\label{eq:DiophantGeneral}
      2\a_{ii}+\sum_{\genfrac{}{}{0pt}{2}{j=1}{j\neq i}}^f\a_{ij}=\ell_{i},
    \end{equation}
    where $1\leq i\leq f$. It is a system of linear Diophantine equations and they have zero, one or infinitely many solutions (counting both positive and negative ones). In this case, the numerical coefficients for the variables $\a_{ij}$ are  simple  so we can easily guess a solution for $i=1$ in~\eqref{eq:DiophantGeneral} (say $\a_{12}=\a_{13}=\hdots=\a_{1\ f-1}=0$ and so $\a_{1f}=\ell_1-2\a_{11}$ which has infinitely many solutions). It also means that there are finitely many nonnegative solutions $\a_{ij}$ we are looking for. Let us prove the second statement first and polynomially upper bound the number of non-negative solutions. By setting $\ell=\max_i{\ell_i}$, the first row of~\eqref{eq:DiophantGeneral} becomes
    \begin{equation}\label{eq:DiophantGeneral1row}
      \a_{12}+\a_{13}+\dots+\a_{1f}=\ell-2\a_{11}.
    \end{equation}
    Let us also set $a_{11}=0$ for the moment. The task of finding the number of nonnegative solutions resembles the problem of finding the number of $k$-compositions of $\ell$ (see Definition~\ref{def:compositions}). Indeed, if we take $1\leq k\leq f-1$ then we only need to know how to distribute $\binom{\ell-1}{k-1}$ $k$-composition in $(f-1)$ `slots'. Or, put differently, how many ways we can pad a $k$-composition with zeros. This equals to $\binom{f-1}{k}\binom{\ell-1}{k-1}$ and by summing over all $k$-compositions we get
    \begin{equation}\label{eq:paddedCompositions}
      \sum_{k=1}^{f-1}\binom{f-1}{k}\binom{\ell-1}{k-1}=\binom{f-2+\ell}{f-2}.
    \end{equation}
    We have to add all possible $\ell-2\a_{11}$ on the RHS of~\eqref{eq:DiophantGeneral1row}. Since $\a_{11}\in\bbZ_{\leq0}$, the parameter $\ell$ decreases by the multiples of two so we will need to distinguish between odd and even $\ell$ and sum only the appropriate ones. But we are interested in an upper bound so let's pretend $\ell$ can be any positive integer and just calculate
    \begin{equation}\label{eq:eq:paddedCompositionsAllEll}
      \sum_{\tilde\ell=0}^\ell\binom{f-2+\tilde\ell}{f-2}={1+\ell\over f-1}\binom{f-1+\ell}{f-2}
      ={1\over(f-1)!}\prod_{j=1}^{f-1}(\ell+j)={1\over(f-1)!}p(\ell),
    \end{equation}
    where $p(\ell)$ is a polynomial of degree $f-1$. There is $f$ linear equations in system~\eqref{eq:DiophantGeneral}. Each consecutive equation has one independent variable less than the previous one but if we assume for a while that all $f$ equations contain~$f$ independent variables then the number of solutions is upper bounded by a polynomial of degree $(f-1)^f$.

    How do we obtain all nonnegative solutions in a systematic way?  Since $0\leq\a_{ii}\leq\lfloor\ell_i/2\rfloor$ and $0\leq\a_{ij}\leq\min{\{\ell_i,\ell_j\}}$ for $i\neq j$ we start in the first equation of system~\eqref{eq:DiophantGeneral} ($i=1$) by
    fixing the lowest possible values of $\a_{11},\a_{12}$ up to $\a_{1\ f-2}$ ($\a_{11}=\a_{12}=\hdots\a_{1\ f-2}=0$) and simply list all the admissible $\a_{1j}$'s for $f-2<j\leq f$. We continue by increasing $\a_{1\ f-2}$ by one and repeat the procedure until we hit $\min{[\ell_1,\ell_2]}$. Then we repeat the whole process for $\a_{1\ f-3}$ all the way to $\a_{11}$. In the next step we move to~\eqref{eq:DiophantGeneral} for $i=2$ by inserting all found $f$-tuples $(\a_{1j})_{1\leq j\leq f}$ and repeat the procedure until we find all admissible solutions in the second row. As a result we obtain an $(2f-1)$-tuple $(\a_{1j},\a_{2j'})_{1\leq j\leq f,2\leq j'\leq f}$. As it is clear from~\eqref{eq:DiophantGeneral} the second row `inherits' one variable ($\a_{12}$) from the first row. We continue in a similar fashion for all $i$. The result is a complete list of $f(f+1)/2$-tuples $\ba$ thus solving Eqs.~(\ref{eq:DiophantGeneral}).
\end{proof}
\begin{rem}
    Note that the number of nonnegative solutions is gruesomely overestimated. One can convert the second part of the proof of Theorem~\ref{thm:IsserlisRefined} into a program that systematically finds all the solutions. Essentially, if it takes one time step to find the first solution then all the $t$ solutions can be find in $t$ time steps where $t$ increases polynomially. There is really no need for solving a Diophantine system by some sophisticated number-theoretic methods since due to the simple form of~(\ref{eq:DiophantGeneral}) we get \emph{all} nonnegative solutions by inspection and only need to list them (i.e., save them to the memory).

    A related problem is how to get a closed form for the number of nonnegative solutions for given $\ell_i$ and $f$ without actually listing and counting the solutions. This is a  nontrivial task attempted by the author in a very special case of $\ell_i=\ell$  and $f=4$~\cite{bradler2016dio} (the chosen values have no relevance to perturbative QFT). It is a problem studied by the Ehrhart  theory~\cite{ehrhart1962} where the interior points of convex lattice polyhedra are counted~\cite{beck2007computing}. The Ehrhart theory does not provide a closed expression for the number of lattice points, however, and, typically, computer algebra systems are used to count the interior points. The connection to convex geometry comes from rewriting system~\eqref{eq:DiophantGeneral} as a system of inequalities
    $$
    \sum_{\genfrac{}{}{0pt}{2}{j=1}{j\neq i}}^f\a_{ij}\leq\ell_{i}
    $$
    which is an algebraic definition of a convex polyhedron.
\end{rem}
To get the conformation multiplicity factor $\mu_{\ba}$ in Theorem~\ref{thm:IsserlisRefined} we present an auxiliary lemma.
\begin{lem}\label{lem:fcnCardinality}
  Let $S$ and $T$ be discrete sets where $|S|=s,|T|=t$. Then there is
  \begin{equation}\label{eq:fcnCardinality}
    \binom{s}{n}\,t\times\ldots\times(t-n+1)
  \end{equation}
  bijective functions $f:X\mapsto Y$ where $X\subset S,Y\subset T$ and $|X|=|Y|=n$ where $0<n\leq\min{\{s,t\}}$. For $n=0$ we set Eq.~(\ref{eq:fcnCardinality}) to one.
\end{lem}
\begin{rem}
  If $s=t=n$ then~(\ref{eq:fcnCardinality}) becomes $n!$ which is  known to be the number of bijections from a set to itself (the number of permutations). To make the notation more concise in following text we will use the definition of the falling factorial: $[m]_n\df m\times\ldots\times (m-n+1)$. So Eq.~(\ref{eq:fcnCardinality}) becomes $\binom{s}{n}[t]_n$.
\end{rem}
\begin{proof}
  The coefficient $\binom{s}{n}$ is simply a number of all possible domains $X\subset S$ whose cardinality is $n$. Then for every domain $X$ there is $\binom{t}{n}n!$ codomains $Y\subset T$ of the same cardinality. The  coefficient $n!$ comes from the number of permutations within each of $\binom{t}{n}$ codomains $Y$.
\end{proof}
\begin{thm}\label{thm:Multiplicity}
  The multiplicity factor $\mu_{\ba}$ for a given conformation exponent $\ba$ reads
  \begin{equation}\label{eq:multFactor}
    \mu_{\ba}=\G\,\prod_{\genfrac{}{}{0pt}{2}{i,j=1}{i<j}}^{f}\g_{ij},
  \end{equation}
  where
  \begin{subequations}\label{eq:Gammas}
    \begin{align}
      \G & = \prod_{i=1}^f\binom{\ell_i}{2\a_{ii}}\prod_{i=1}^f(2\a_{ii}-1)!!,\label{eq:GammasGAMMA}\\
      \g_{ij} & = \binom{\ell_i-2\a_{ii}-\sum_{m=1}^{i-1}\a_{mi}-\sum_{n=i+1}^{j-1}\nolimits\a_{in}}{\a_{ij}}\big[\ell_j-2\a_{jj}-\sum_{m=1}^{i-1}\a_{mj}\big]_{\a_{ij}},\label{eq:Gammasgaij}
    \end{align}
  \end{subequations}
  where $[\bullet]_{\a_{ij}}$ denotes the falling factorial introduced in the remark below Lemma~\ref{lem:fcnCardinality}.
\end{thm}
\begin{proof}
  We will use repeatedly Lemma~\ref{lem:fcnCardinality} by identifying $\lan ij\ran_0^{\a_{ij}}$ for $i\neq j$ from Theorem~\ref{thm:IsserlisRefined} by setting:
  \begin{align}
    n & =\a_{ij}, \\
    s & =\ell_i,\\
    t & =\ell_j.
  \end{align}
  Set $S$ is the set of all $i$'s and $T$ is the set of all $j$'s. However, if we calculate more than one two-point Green's function (like in Theorem~\ref{thm:IsserlisRefined}) we have to take into account the fact that the sets $S$ and $T$ might have shrunk. This depends on whether in the preceding Green's function $\lan kl\ran_0^{\a_{kl}}$ we had $k=i$ or $l=j$. The strategy we will follow here is to first count the multiplicity of $\lan ii\ran_0^{\a_{ii}}$ as they are independent (meaning non-overlapping for different $i$'s) and then the multiplicities of $\lan ij\ran_0^{\a_{ij}}$. In that case the derivation follows the lexicographic ordering, Def.~\ref{def:lexi}, which is a handy tool here.

  The multiplicity of  $\lan ii\ran_0^{\a_{ii}}$ immediately follows from Isserlis' theorem. Looking at Eq.~(\ref{eq:Isserlis}) we see that if $x_i=x_j$ there will be $(2m-1)!!$ identical products. Hence ($m$ is $\a_{ii}$ here)
  \begin{equation}\label{eq:selfMultiplicity}
    \g_{ii}=\binom{\ell_i}{2\a_{ii}}(2\a_{ii}-1)!!,
  \end{equation}
  as follows from Eq.~(\ref{eq:fcnCardinality})\footnote{Indeed, the double factorial follows either from Theorem~\ref{thm:Isserlis} as stated or from the counting done in Lemma~\ref{lem:fcnCardinality}.}. The factor of two in the binomial `denominator' accounts for the two $i$'s in  $\lan ii\ran_0$. Repeating this procedure for all $i$ we get $\G$ in Eq.~(\ref{eq:GammasGAMMA}). To get $\g_{ij}$ we then have to keep track of the set cardinality and this is the point where the used lexicographic order becomes important. We follow the ordered set of $ij$ where $1\leq i<j\leq f$ (note the sharp inequality). Hence, $ij=(12,\dots,1f,23,\dots,ff)$ and we find all $\g_{ij}$ where $i\neq j$. The first one is $\g_{12}$ and so we set $n=\a_{12}$ (using the notation of Lemma~\ref{lem:fcnCardinality}) and notice that the cardinality of the discrete set of `ones' is decreased by those contributing to the previously analyzed case $i=j$, namely $\g_{11}$. So $s=\ell_1-2\a_{11}$. Similarly, the cardinality of the `twos' is decreased by the number of elements contributing to $\g_{22}$. So $t=\ell_2-2\a_{22}$ and $\g_{12}=\binom{\ell_1-2\a_{11}}{\a_{12}}[\ell_2-2\a_{22}]_{\a_{12}}$ follows. At this point we can proceed and construct a generic $\g_{ij}$. The `source pool' of $i$'s will be depleted by three contributions: $2\a_{ii},\sum_{m=1}^{i-1}\nolimits\a_{mi}$ and $\sum_{n=i+1}^{j-1}\nolimits\a_{in}$. The bounds of the sums are chosen such that only the preceding contributions to the sought $\g_{ij}$ are accounted for. Similarly, the `target pool' of $j$'s is depleted by $2\a_{jj}$ and $\sum_{m=1}^{i-1}\a_{mj}$, but the sum bounds only choose those contributions preceding  $\g_{ij}$ (using the lexicographic ordering). We notice an interesting asymmetry, where in the falling factorial part of~\eqref{eq:Gammasgaij} there is only one sum. This is because all (negative) contributions $\g_{jl}$ come after any $\g_{ij}$ given the lexicographic ordering and $\g_{jj}$ was taken care of separately.
\end{proof}
\begin{rem}
    We again emphasize the importance of the lexicographic ordering by which we followed the construction of $\g_{ij}$'s. Other orderings are certainly plausible and have to provide the same $\mu_{\ba}$, Eq.~\eqref{eq:multFactor}. But we believe the lexicographic ordering and the construction based on it is the most intuitive one.
\end{rem}
Having a decomposition into a product of two-point Green's functions obtained through Theorem~\ref{thm:IsserlisRefined} we can proceed as usual. The position-based Feynman diagrams can be immediately read off from the products of Green's functions. However, even different Diophantine solutions (different Green's products) lead to the same Feynman diagram, that is, to the same physical processes. This is typical for real scalar theories since when we permute the internal degrees of freedom it is often the same physical process (the same integral in the $S$-matrix). However, when charges are involved, like for complex fields, the internal lines are directed and we may need to distinguish among the Diophantine solutions. For more, see Sections~\ref{sec:apps} and~\ref{sec:beyondscalars}. At this point it is advantageous to express our results in the language of graph theory which leads to Feynman diagrams. This is just an advantageous way of presenting the results and not a starting point as it typically  is in QFT. The connection between Feynman diagrams and graph theory is well-documented~\cite{nakanishi1971graph} and here we recall a few basic concepts and results~\cite{kreher1998combinatorial,harary2014graphical}.
\begin{defi}\label{def:graph}
  A labeled undirected graph $G=(V,E)$ is a set of vertices $V$ and a set $E$ of unordered pairs of elements of $V$ called edges. A graph is called simple if no loops and no multiple edges connecting the same pair of vertices are allowed. Graphs allowing both loops and multiple edges are called pseudographs.
\end{defi}
The following definition of graph isomorphism is split into two: the standard one~\cite{kreher1998combinatorial} and a version upgraded by a requirement on the labeling, see~\cite[Chapter~1]{harary2014graphical}. We will need both definitions.
\begin{defi}\label{def:graphIsoAuto}
A type-1 graph isomorphism $\euI_1(G_1,G_2)$  is a bijection $\euI_1:V_1\mapsto V_2$  such that $\euI_1(E_1)=E_2$. A type-2 graph isomorphism $\euI_2(G_1,G_2)$  is a bijection $\euI_2:V_1\mapsto V_2$  such that $\euI_2(E_1)=E_2$ while also preserving  the labeling. An automorphism of a labeled graph $G$ is defined as $\euA(G)=\euI_2(G,G)$. The set of all automorphisms of $G$ forms the group $\mathsf{Aut}{[G]}$.
\end{defi}
If we restrict ourselves to boson scalar theories then a position-based Feynman diagram is a labeled undirected pseudograph $G=(V,E)$ (we will omit the prefix pseudo- and all other adjectives from now on unless we need them to avoid confusion). The main observation here is that the conformation exponent $\ba=(\a_{ij})_{1\leq i\leq j\leq f}$ from Theorem~\ref{thm:IsserlisRefined} generates $G$ where $|V|=f$ and $0<\a_{ij}\in E$. Let us summarize in
\begin{defi}
  We call $G_{\ba}$ an $\ba$-generated labeled undirected graph if every $\a_{ij}\neq0$ (from Theorem~\ref{thm:IsserlisRefined}) denotes $\a_{ij}$ edges connecting two vertices $i$ and $j$. We split the edges and vertices into external and internal ones according to the details of the studied boson theory, calculated perturbative order  and the number of interacting fields and write $G_{\ba}=(V_\mathrm{int}\cup V_{\mathrm{ext}},E_\mathrm{int}\cup E_{\mathrm{ext}})$. We further define the internal subgraph of $G_{\ba}$ as $G_{\ba,\mathrm{int}}\df(V_\mathrm{int},E_\mathrm{int})$.
\end{defi}
The internal vertices are the internal degrees of freedom that are integrated over in Lagrangian~\eqref{eq:GreenInteracting} ($\int\diff{d}{\boldsymbol{z}}$). The external vertices represent the interacting fields ($(x_i)_{i=1}^k$) and are connected to the internal vertices by external edges. The internal edges connect the internal vertices. Finally, only connected graphs are considered.
\begin{rem}
  A symmetric adjacency matrix is a way of encoding an undirected graph. If we take an upper-triangular part of the matrix and flatten it row by row we get our conformation exponent~$\ba$.
\end{rem}
\begin{thm}{\cite{harary1967number}}\label{thm:labelings}
  The number of labelings of a given graph~$G$ with $p$ vertices is
  \begin{equation}\label{eq:numofLabeling}
    l(G)={p!\over|\mathsf{Aut}{[G]}|}.
  \end{equation}
\end{thm}
Theorem~\ref{thm:labelings} answers the question of how many  labeled type-1 isomorphic graphs to $G$ there are.
\begin{lem}\label{lem:FeynmanMulti}
   Let $G_{\ba}$  be an $\ba$-generated labeled undirected graph where $I\df|V_\mathrm{int}|,E\df|V_\mathrm{ext}|$ and $(z_i)_{i=1}^I\in V_\mathrm{int}$ together with $(x_i)_{i=1}^E\in V_\mathrm{ext}$. Given a partition of $V_{\mathrm{ext}}$ into  $J\leq I$ disjoint sets $V_{\mathrm{ext}}^{(j)}$ of cardinality $E^{j}\df|V_{\mathrm{ext}}^{(j)}|$
   $$
   V_{\mathrm{ext}}=\bigcup_{j=1}^{J} V_{\mathrm{ext}}^{(j)},
   $$
   there is
    \begin{equation}\label{eq:FeynmanMulti}
       \xi_{\ba}=\prod_{i=1}^J\binom{E-\sum_{j=1}^{i-1}E^{j-1}}{E^i}{I!\over|\mathsf{Aut}{[G_{\ba,\mathrm{int}}]}|}
    \end{equation}
   Feynman diagrams representing the same physical interaction. We set $E^0=0$.
\end{lem}
  \begin{figure}[h]
   \resizebox{10cm}{!}{\includegraphics{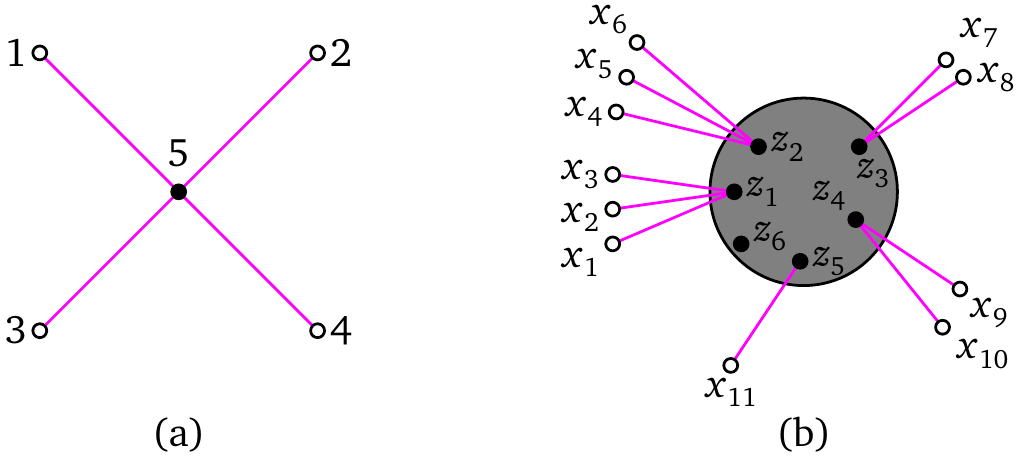}}
    \caption{The illustration of what is counted by Lemma~\ref{lem:FeynmanMulti}. (a) There are four external ($E=4$) and one internal vertices ($I=1$). Hence $J=1$ and from~\eqref{eq:FeynmanMulti} we get $\xi_{\ba}=\binom{4}{4}{1!\over1}=1$ as expected since all 24 permutations of the external vertices correspond to the same Diophantine solution. (b) We count $E=11, I=6$ and $J=5$ where $E^1=E^2=3,E^3=E^4=2$ and $E^5=1$.  By plugging all into~\eqref{eq:FeynmanMulti} we get $\xi_{\ba}=\binom{11}{3}\binom{8}{3}\binom{5}{2}\binom{3}{2}\binom{1}{1}{6!\over|\mathsf{Aut}{[G_{\ba,\mathrm{int}}]}|}$, where $\mathsf{Aut}{[G_{\ba,\mathrm{int}}]}$ depends on the internal subgraph structure represented by the grey circle.}
    \label{fig:multipl}
   \end{figure}
\begin{rem}[\bfseries{important}]
  Let us emphasize what we do and do not count in this lemma. In Fig.~\ref{fig:multipl}~(a), there is an example of a simple Feynman diagram with 4! ways of permuting the external vertices. But at the level of a system of Diophantine equations, all 4! possibilities corresponds to a single nonnegative solution and its  multiplicities were calculated in Theorem~\ref{thm:Multiplicity}.
   In Fig.~\ref{fig:multipl} (b), there is a generic $\ba$-generated graph where the grey circle represents an arbitrary relation among the six internal vertices $z_i$. There can be multiple edges and loops and internal vertices need not be connected to any external vertex. The external vertices are labeled by $x_i$ and a swap of $x_i$ with different $z_i$ corresponds to a different solution of a Diophantine system. All such possibilities are then enumerated by $\xi_{\ba}$ in~\eqref{eq:FeynmanMulti}.
\end{rem}
\begin{proof}[Proof of Lemma~\ref{lem:FeynmanMulti}]
  The graph $G_{\ba}$ is $\ba$-generated and so we assume that if an internal vertex connects to one or more external vertices their number is fixed. Then, the product of binomials $\prod_{i=1}^J\binom{E-\sum_{j=1}^{i-1}E^{j-1}}{E^i}$ counts the number of ways the external vertices can be connected to the internal ones. The $i$-th binomial `numerator' is depleted by the external vertices already connected to the internal ones and there are naturally many ways (depending on the order of $E^j$) leading to the same result. The coefficient $I!$ in~\eqref{eq:FeynmanMulti} comes from Theorem~\ref{thm:labelings} and it is a permutation of all internal vertices. We now introduce two procedures named \emph{amputation} (not to be confused with the equally called procedure for removing infinities in the momentum-based Feynman diagrams!) and \emph{grafting} in order to determine $|\mathsf{Aut}{[G_{\ba,\mathrm{int}}]}|$. For a given $G_{\ba}$ we first remove all the external vertices but keep the external edges `freely floating'. This is called amputation and at this point it is not a graph. We make it a graph again by creating loops of out the amputated external edges. This step will be called grafting and the whole procedure is illustrated in Fig.~\ref{fig:ampgraft}. The newly created loops hold the information about the internal subgraph symmetry and thus its automorphism group (using Definition~\ref{def:graphIsoAuto}~type-2 isomorphism). Then, ${I!\over|\mathsf{Aut}{[G_{\ba,\mathrm{int}}]}|}$ is the total number of internal graphs with different $\ba$ corresponding to different solutions of the Diophantine system leading to the same graph.
\end{proof}
  \begin{figure}[h]
   \resizebox{14.7cm}{!}{\includegraphics{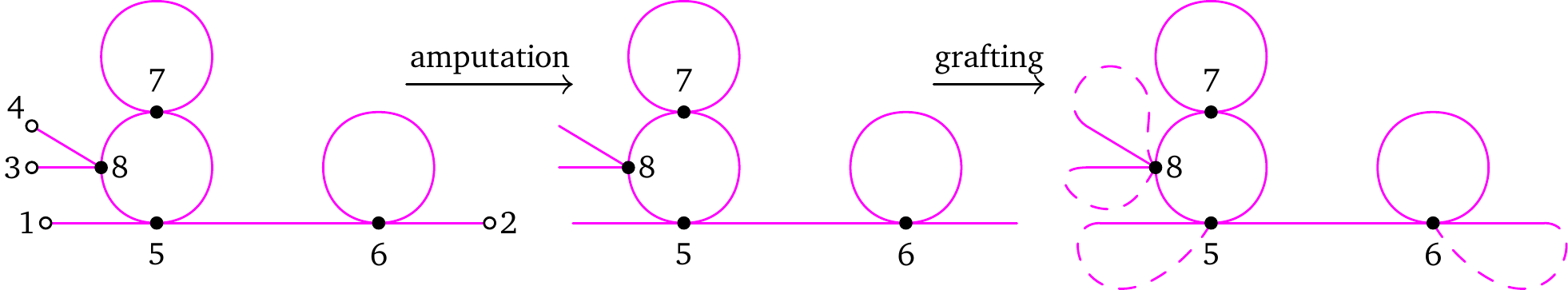}}
    \caption{The procedure of amputation removes all the external vertices (white circles) followed by grafting where the external `stubs' are made into loops (the dashed reconnections).}
    \label{fig:ampgraft}
   \end{figure}
\begin{rem}
  Note that a grafted graph is not a vacuum Feynman diagram. This can be seen in Fig.~\ref{fig:ampgraft} on the right where some vertices have different vertex orders (the number of edges meeting there). It is just a pseudograph we have created to help us calculate the overall multiplicity factor of the graph we started with. Even though the graph automorphism order is always studied when dealing with Feynman diagrams (explicitly like in~\cite{abdesselam2003feynman} or implicitly like everywhere else), our approach invokes it at a different stage of the calculations and for grafted graphs compared to other studies.
\end{rem}
How do we calculate $\xi_{\ba}$? There are two different routes. We can follow Theorem~\ref{thm:labelings} and that includes the calculation of $|\mathsf{Aut}{[G_{\ba,\mathrm{int}}]}|$. For moderate graphs, it can be done by realizing that the whole list of Diophantine solutions $\ba$ splits into several classes of physically indistinguishable processes (the same Feynman diagrams). We take  \emph{any} representative of the class we are interested in, process it according to the previous lemma and calculate $|\mathsf{Aut}{[G_{\ba,\mathrm{int}}]}|$. The problem of graph automorphism is not known to be tractable and is closely related to the well-known graph isomorphism membership problem~\cite{hoffmann,kreher1998combinatorial}. However, practically efficient algorithms are known and implemented (for pseudographs and discrete mathematics in general, SAGE~\cite{sagemath} is a great tool).  Even though almost all finite graphs possess no symmetry (that it, a trivial $\mathsf{Aut}{[G_{\ba,\mathrm{int}}]}$, as proved in~\cite{erdos1963asymmetric} for simple graphs), we need the exact counting.

To check the graph automorphism order we can take the second route to get $\xi_{\ba}$. By solving a Diophantine system we get all admissible  conformation exponents  $\ba$ and within each class there is obviously never the same $\ba$ more than once.  In other words, the graph automorphism order is already accounted for and $\xi_{\ba}$ is nothing else than the cardinality of the class we are interested in. We did not get away with the hardness of the problem, though. In order to compute the cardinality we have to take a representative and check how many graphs from the list of all $\ba$'s  it is isomorphic to (using Definition~\ref{def:graphIsoAuto},~type-1 isomorphism). This approach is most likely  harder from the computational complexity viewpoint because  the graph isomorphism membership algorithm has to be invoked many times and it is computationally equivalent to the graph automorphism order~\cite{hoffmann}. On the other hand, recall that the Diophantine system of solutions grows polynomially and, mainly, many solutions are computationally easy to disprove to be isomorphic to the studied representative by checking whether the graph is (dis)connected, by counting  the edges or checking any other graph invariant~\cite{kreher1998combinatorial} that is computationally cheap to implement on the level of the conformation exponent $\ba$ (to name a few more it is the number of loops, the ways the external edges are connected to the internal vertices etc.). But to the author's knowledge there is no guarantee that two different classes always differ at least in one invariant (for any bosonic theory). After  finding $\xi_{\ba}$, we can immediately  obtain the graph automorphism  order by counting $E,E^j$ and $I$  and inserting it into~Eq.~\eqref{eq:FeynmanMulti}. Both routes will be illustrated in the next section.

Our  hope is that the structure of a convex lattice polyhedron where all nonnegative solutions of a Diophantine system live will provide yet another way of counting the graph automorphism order, this time without the need to run computationally costly algorithms. We will briefly return to this open problem in Section~\ref{sec:discussions}.

\section{Applications and examples}\label{sec:apps}

Several explicit examples are worked out in the next section. One can promptly (that is, in linear time in the number of Diophantine solutions) processed the solutions of the Diophantine system and keep only the diagrams of interest such as all connected diagrams or all vacuum diagrams and so on. It is well-known that the vacuum contributions in the numerator of Eq.~\eqref{eq:GreenInteracting} cancel with the free $S$-matrix in the denominator~\cite{schwartz2014quantum} and the cancellation can be done at the level of Diophantine solutions. Again, we illustrate it in the next section.
%

%

\begin{exa}[$\phi^4$ for $k=2$, second perturbative order in detail]
  Let's calculate the third term from the expanded numerator of~Eq.~(\ref{eq:GreenInteracting}). It is proportional to
  \begin{equation}\label{eq:phi4expansionk2}
    \int\diff{4}{\boldsymbol{z}}\bra{0_M}\mathsf{T}\{\phi(x_1)\phi(x_2)\phi(z_1)^4\phi(z_2)^4\}\ket{0_M}
    =  \int\diff{4}{\boldsymbol{z}} \lan1^{1}2^{1}3^{4}4^{4}\ran_0.
  \end{equation}
  Upon setting $f=4$, Eq.~\eqref{eq:GreensFcns} becomes
  \begin{equation}\label{eq:GreensSplitf4}
    \lan1^{\ell_1}2^{\ell_2}3^{\ell_3}4^{\ell_4}\ran_0
    \equiv\lan\underbrace{1\dots1}_{\ell_1}\underbrace{2\dots2}_{\ell_2}\underbrace{3\dots3}_{\ell_3}\underbrace{4\dots4}_{\ell_4}\ran_0,
  \end{equation}
  Next, we solve~\eqref{eq:DiophantGeneral} that took the following form:
  \begin{subequations}\label{eq:Diophantinef4}
  \begin{align}
    2\a_{11}+\a_{12}+\a_{13}+\a_{14} & = \ell_1,\label{eq:Diophantine1} \\
    \a_{12}+2\a_{22}+\a_{23}+\a_{24} & = \ell_2,\label{eq:Diophantine2} \\
    \a_{13}+\a_{23}+2\a_{33}+\a_{34} & = \ell_3, \\
    \a_{14}+\a_{24}+\a_{34}+2\a_{44} & = \ell_4,\label{eq:Diophantine4}
  \end{align}
  \end{subequations}
  where $\ell_1=\ell_3=1$ and $\ell_3=\ell_4=4$. The result is eleven conformation exponents
  \begin{equation}\label{eq:ell1144}
   (\ba_c)_{c=1}^{11}=\left\{
    \begin{array}{cccccccccc}
     (0 ,& 0 ,& 0 ,& 1 ,& 0 ,& 0 ,& 1 ,& 1 ,& 2 ,& 0) \\
     (0 ,& 0 ,& 0 ,& 1 ,& 0 ,& 1 ,& 0 ,& 0 ,& 3 ,& 0) \\
     (0 ,& 0 ,& 0 ,& 1 ,& 0 ,& 1 ,& 0 ,& 1 ,& 1 ,& 1) \\
     (0 ,& 0 ,& 1 ,& 0 ,& 0 ,& 1 ,& 0 ,& 0 ,& 2 ,& 1) \\
     (0 ,& 0 ,& 1 ,& 0 ,& 0 ,& 0 ,& 1 ,& 0 ,& 3 ,& 0) \\
     (0 ,& 0 ,& 1 ,& 0 ,& 0 ,& 0 ,& 1 ,& 1 ,& 1 ,& 1) \\
     (0 ,& 0 ,& 0 ,& 1 ,& 0 ,& 0 ,& 1 ,& 2 ,& 0 ,& 1) \\
     (0 ,& 1 ,& 0 ,& 0 ,& 0 ,& 0 ,& 0 ,& 0 ,& 4 ,& 0) \\
     (0 ,& 1 ,& 0 ,& 0 ,& 0 ,& 0 ,& 0 ,& 1 ,& 2 ,& 1) \\
     (0 ,& 1 ,& 0 ,& 0 ,& 0 ,& 0 ,& 0 ,& 2 ,& 0 ,& 2) \\
     (0 ,& 0 ,& 1 ,& 0 ,& 0 ,& 1 ,& 0 ,& 1 ,& 0 ,& 2) \\
    \end{array}
    \right\}
  \end{equation}
  whose lexicographic ordering is
  $$
  \ba=(\a_{11},\a_{12},\a_{13},\a_{14},\a_{22},\a_{23},\a_{24},\a_{33},\a_{34},\a_{44}).
  $$
  The solutions and the values of $\ell_i$ are inserted into
  \begin{subequations}\label{eq:Gammasf4}
    \begin{align}
      \G & = \binom{\ell_1}{2\a_{11}}\binom{\ell_2}{2\a_{22}}\binom{\ell_3}{2\a_{33}}\binom{\ell_4}{2\a_{44}}\prod_{i=1}^4(2\a_{ii}-1)!!,\\
      \g_{12} & = \binom{\ell_1-2\a_{11}}{\a_{12}}[\ell_2-2\a_{22}]_{\a_{12}},\label{eq:Gammasga12}\\
      \g_{13} & = \binom{\ell_1-2\a_{11}-\a_{12}}{\a_{13}}[\ell_3-2\a_{33}]_{\a_{13}},\label{eq:Gammasga13}\\
      \g_{14} & = \binom{\ell_1-2\a_{11}-\a_{12}-\a_{13}}{\a_{14}}[\ell_4-2\a_{44}]_{\a_{14}},\\
      \g_{23} & = \binom{\ell_2-2\a_{22}-\a_{12}}{\a_{23}}[\ell_3-2\a_{33}-\a_{13}]_{\a_{23}}, \\
      \g_{24} & = \binom{\ell_2-2\a_{22}-\a_{12}-\a_{23}}{\a_{24}}[\ell_4-2\a_{44}-\a_{14}]_{\a_{24}},\\
      \g_{34} & = \binom{\ell_3-2\a_{33}-\a_{13}-\a_{23}}{\a_{34}}[\ell_4-2\a_{44}-\a_{14}-\a_{24}]_{\a_{34}}
    \end{align}
  \end{subequations}
  to get  Eq.~\eqref{eq:multFactor} whose explicit form reads
  \begin{equation}\label{eq:multFactorf4}
    \mu_{\ba}=\G\,\g_{12}\g_{13}\g_{14}\g_{23}\g_{24}\g_{34}.
  \end{equation}
  \begin{figure}[t]
   \resizebox{11.3cm}{!}{\includegraphics{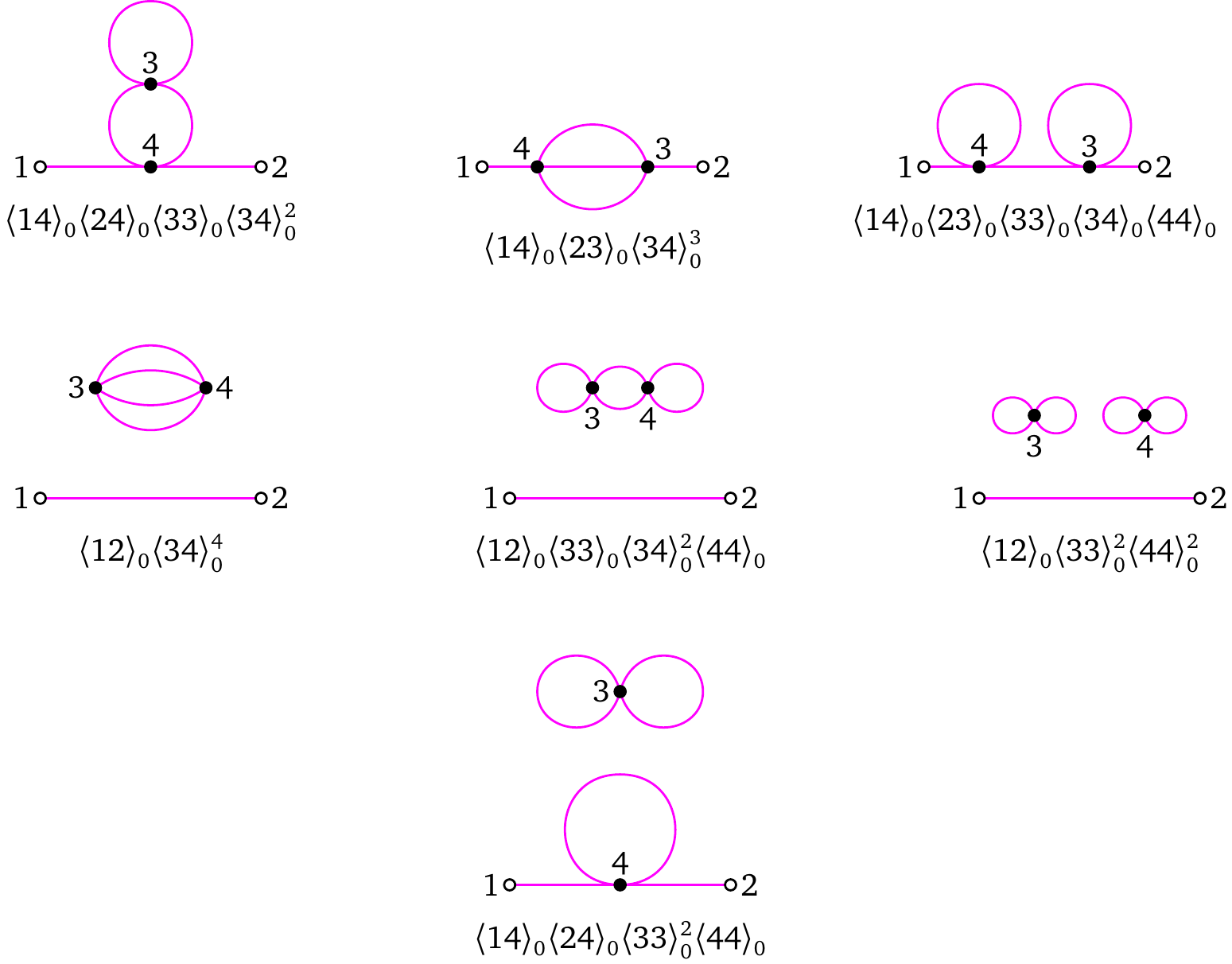}}
    \caption{Seven unique conformations (for the used convention see Definition~\ref{def:coformations}) and their Feynman diagrams out of eleven conformations~Eq.~(\ref{eq:phi4expansionk2}) splits into. The black dots are the internal coordinates $z_1$ and $z_2$, the white circles are the external fields $x_1$ and $x_2$. }
    \label{fig:feynmanphi4}
   \end{figure}
  We obtain
  \begin{equation}\label{eq:muell1144}
    \mu_{\ba_c}=(144, 96, 144, 144, 96, 144, 36 , 24, 72, 9, 36)
  \end{equation}
  and perform a highly non-trivial check of both Theorems~\ref{thm:IsserlisRefined} and~\ref{thm:Multiplicity}:
  $$
  \sum_{c=1}^{11}\mu_{\ba_c}=\Big(\sum_{i=1}^{4}\ell_i-1\Big)!!=9!!.
  $$
  We obtained eleven diagrams but not all are physically distinguishable. To account for them  we use Theorem~\ref{thm:Multiplicity} (valid for connected diagrams which are physically most interesting). The unique diagrams and their conformations are depicted in Fig.~\ref{fig:feynmanphi4}. In the first row we put all three connected diagrams (the first three rows in~\eqref{eq:ell1144}) with  multiplicity two which we will now verify. $E=2$ and $I=2$ hold for all diagrams. For the left upper diagram we find $E^1=2$ while for the other two we get $E^1=E^2=1$. Then, following the procedure of amputation and grafting, we create new graphs where the upper left one has no symmetry (and so $|\mathsf{Aut}{[G_{\ba_1,\mathrm{int}}]}|=1$) while the other two are identical upon swapping $3\leftrightharpoons4$ (hence $|\mathsf{Aut}{[G_{\ba_{2,3},\mathrm{int}}]}|=2$). So we get
    \begin{subequations}
        \begin{align}
          \xi_{\ba_1} & = \binom{2}{2}{2!\over1}=2,\\
          \xi_{\ba_2} & = \binom{2}{1}\binom{1}{1}{2!\over2}=2,\\
          \xi_{\ba_3} & = \binom{2}{1}\binom{1}{1}{2!\over2}=2
        \end{align}
    \end{subequations}
  and indeed $\ba_{4,5,6}$ in~\eqref{eq:ell1144} are the other class members for $\ba_{1,2,3}$ (in this order). We just explored the two routes of calculating $\xi_{\ba}$ as described below Lemma~\ref{lem:FeynmanMulti}. A final check is the overall multiplicity of a given Feynman diagram: $\nu_{\ba_i}=\mu_{\ba_i}\xi_{\ba_i}=(288,192,288)$ for $i=1,2,3$ in accord with~\cite[Table~II]{kleinert2000recursive}.
\end{exa}
\begin{rem}
    By solving~\eqref{eq:Diophantinef4} for $\ell_i=2$ we get 17  conformations  $\lan1^{2}2^{2}3^{2}4^{2}\ran_0$  and find
    \begin{equation}\label{eq:multiell2222}
      \mu_{\ba_c}=(2, 2, 2, 1, 8, 4, 2, 8, 16, 4, 2, 8, 16, 16, 8, 4, 2)
    \end{equation}
    from~\eqref{eq:multFactorf4}. A calculation performed in~\cite{bradler2016dio} leads to a closed expression for the number of non-negative solutions of~\eqref{eq:Diophantinef4} (for $\ell_i=\ell$ even) to be
  \begin{equation}\label{eq:ellEven}
    \mathsf{e}(\ell)=\frac{1}{576} (\ell+2) (\ell+4)\big(\ell (\ell+5) (\ell (\ell+4)+12)+72\big).
  \end{equation}
  For $\ell=2$ we get $\mathsf{e}(2)=17$ in accord with the above result. In theory, one can get closed expressions for the number of nonnegative Diophantine solutions following the  methods of Ehrhart theory~\cite{ehrhart1962,beck2007computing}.
\end{rem}
\begin{exa}[$\phi^4$ for $k=4$,  third perturbative order]
  We study
  \begin{equation}\label{eq:phi4expansionk4}
    \int\diff{4}{\boldsymbol{z}}\bra{0_M}\mathsf{T}\{\phi(x_1)\phi(x_2)\phi(x_3)\phi(x_4)\phi(z_1)^4\phi(z_2)^4\phi(z_3)^4\}\ket{0_M}
    =  \int\diff{4}{\boldsymbol{z}} \lan1^{1}2^{1}3^{1}4^{1}5^46^47^4\ran_0.
  \end{equation}
  By solving the corresponding Diophantine system we obtain 960 nonnegative solutions whose conformation exponents will not be listed. After filtering out the disconnect graphs we find in  total 8 classes of connected graphs just like in~\cite[Table~III,~$p=3$]{kleinert2000recursive}. Our results completely agree with the multiplicities of all 8 Feynman diagrams. Let's provide two class examples whose representatives are
    \begin{align}
        \ba_1&=(0,0,0,0,0,0,1,0,0,0,0,0,1,0,0,1,0,0,0,1,0,0,0,1,1,1,1,0), \\
        \ba_2&=(0,0,0,0,0,0,1,0,0,0,1,0,0,0,0,1,0,0,0,1,0,0,0,0,1,1,2,0),
    \end{align}
  where $\mu_{\ba_1}=6912,\mu_{\ba_2}=3456$. For $\ba_1$ we find $E=4,E^1=E^2=2$ and $I=3$. Following Lemma~\ref{lem:FeynmanMulti} we get a grafted graph $G_{\ba_1,\mathrm{int}}$ where $|\mathsf{Aut}{[G_{\ba_1,\mathrm{int}}]}|=2$ leading to
  \begin{equation}
    \xi_{\ba_1}= \binom{4}{2}\binom{2}{2}{3!\over2}=18.
  \end{equation}
  Thus $\nu_{\ba_1}=\mu_{\ba_1}\xi_{\ba_1}=124416$ as it should be. In the second case we have $E=4,E^1=3,E^2=1$ and $I=3,|\mathsf{Aut}{[G_{\ba_2,\mathrm{int}}]}|=1$ and so
  \begin{equation}
    \xi_{\ba_2}= \binom{4}{3}\binom{1}{1}{3!\over1}=24.
  \end{equation}
  Therefore, $\nu_{\ba_2}=\mu_{\ba_2}\xi_{\ba_2}=82944$ again in agreement with~\cite{kleinert2000recursive}.
\end{exa}
\begin{exa}[$\phi^3$ for $k=2$, second perturbative order plus vacuum amplitudes]
  Here we analyze
  \begin{equation}\label{eq:phi3expansionNUM}
    \int\diff{4}{\boldsymbol{z}}\bra{0_M}\mathsf{T}\{\phi(x_1)\phi(x_2)\phi(z_1)^3\phi(z_2)^3\}\ket{0_M}
    =  \int\diff{4}{\boldsymbol{z}} \lan1^{1}2^{1}3^{3}4^{3}\ran_0.
  \end{equation}
  Similarly, the second expansion term of the denominator of~\eqref{eq:GreenInteracting} (the vacuum digrams) is proportional to
  \begin{equation}\label{eq:phi3expansionDENOM}
    \int\diff{4}{\boldsymbol{z}}\bra{0_M}\mathsf{T}\{\phi(z_1)^3\phi(z_2)^3\}\ket{0_M}
    =  \int\diff{4}{\boldsymbol{z}} \lan1^{3}2^{3}\ran_0.
  \end{equation}
  Using Theorem~\ref{thm:IsserlisRefined} we find 8 conformations for Eq.~(\ref{eq:phi3expansionNUM})
  and Theorem~\ref{thm:Multiplicity} yields their multiplicities
  \begin{equation}\label{eq:multiell1133}
      \mu_{\ba_c}=(18, 6, 9, 18, 9, 18, 9, 18).
  \end{equation}
  We check the solution by inspecting
  \begin{equation}\label{eq:checkell1133}
      \sum_{c=1}^{8}\mu_{\ba_c}=105=7!!=\Big(\sum_{i=1}^{4}\ell_i-1\Big)!!,
  \end{equation}
  where the RHS follows from $\ell_1=\ell_2=1$ and $\ell_3=\ell_4=3$.  By considering just connected diagrams, $(\ba_1,\ba_8)$ and $(\ba_4,\ba_6)$ represent the same processes as confirmed by Lemma~\ref{lem:FeynmanMulti} after we do the proper counting, grafting and graph automorphism counting:
  \begin{subequations}
        \begin{align}
          \xi_{\ba_1} & = \binom{2}{2}{2!\over1}=2,\\
          \xi_{\ba_4} & = \binom{2}{1}\binom{1}{1}{2!\over2}=2.
        \end{align}
    \end{subequations}
  Hence  $\nu_{\ba_1}=\nu_{\ba_4}= 36$ as we can find in~\cite[Chapter~7]{schwartz2014quantum}.

  For~\eqref{eq:phi3expansionDENOM} we get only two solutions
  \begin{equation}
    \{\ba_c\}_{c=1}^{2}=\left\{
    \begin{array}{cccccccccc}
     (0 ,& 0 ,& 0 ,& 0 ,& 0 ,& 0 ,& 0 ,& 0 ,& 3 ,& 0) \\
     (0 ,& 0 ,& 0 ,& 0 ,& 0 ,& 0 ,& 0 ,& 1 ,& 1 ,& 1)
    \end{array}
    \right\}
  \end{equation}
  with multiplicities $\mu_{\ba_c}=(6,9)$. As these are vacuum graphs, there is no need to remove the external edges and graft and so following Lemma~\ref{lem:FeynmanMulti} we find that $\xi_{\ba_1}=\xi_{\ba_2}=2!/2=1$. The final answer is then  $\nu_{\ba_1}=6,\nu_{\ba_2}= 9$ in agreement with~\cite[Chapter~7]{schwartz2014quantum}.
\end{exa}

\subsection*{A pair of Unruh-DeWitt detectors}

A realistic Unruh-DeWitt detector~\cite{unruh1976notes,dewitt1979quantum} is described by the following interaction Hamiltonian
\begin{equation}\label{eq:UDWHint}
  H(\tau)=\la w(\tau)(\s^+e^{i\tau\d}+\s^-e^{-i\tau\d})\int\difff{x} f(x)\phi(t,x),
\end{equation}
and
\begin{equation}\label{eq:RealField}
  \phi(t,x)={1\over(2\pi)^3}\int{\difff{k}\over2\om_k}\big(a_ke^{-i(\om_kt-k\cdot x)}+a^\dg_ke^{i(\om_kt-k\cdot x)}\big)
\end{equation}
is a real scalar field where $(\om_k,k)$ is a four-momentum, $w(\tau)$ is a detector's window function ($\tau$ is its proper time and $t(\tau),x(\tau)$ are Minkowski coordinates), $f(x)$ a smearing function, $\s^\pm$ are the detector ladder operators, $\d$ stands for the energy gap and $\la$ for a coupling constant. The evolution operator for observer $A$ reads
\begin{align}\label{eq:2ndOrderExp}
  U_A(\tau_1,\tau_0)\otimes U_B(\tau_2,\tau_0)
  & = \mathsf{T}{\Big\{\exp{\big[-i\int_{\tau_0}^{\tau_1}\dif{\tau'}H_A(\tau')\big]}\exp{\big[-i\int_{\tau_0}^{\tau_2}\dif{\tau''}H_B(\tau'')\big]}\Big\}}.
\end{align}

Out task is to efficiently calculate the following unitary (tensor product and time-ordering implied)
\begin{align}\label{eq:twoUDWdetectors}
U_A(\tau_1,\tau_0)\otimes U_B(\tau_2,\tau_0)&=\mathsf{T}{
\Big\{\int_{\tau_0}^{\tau_1}\dif{\tau'}\int_{\tau_0}^{\tau_2}\dif{\tau''}e^{-i(A_+\sigma_A^++A_-\sigma_A^-+B_+\sigma_B^++B_-\sigma_B^-)}\Big\}} \nn\\
&= \mathsf{T}{
\Big\{\sum_{n=0}^\infty \frac{(-i)^n\la^n}{n!} (A_+\sigma_A^++A_-\sigma_A^-+B_+\sigma_B^++B_-\sigma_B^-)^n\Big\}}
\end{align}
for any $n$, no matter how large, and for any input atomic state. We denoted
\begin{align}
    A_+ & = \int_{\tau_0}^{\tau_1}\dif{\tau'}w(\tau')e^{i\tau'\d}\phi(t'(\tau')), \\
    B_+ & = \int_{\tau_0}^{\tau_2}\dif{\tau''}w(\tau'')e^{i\tau''\d}\phi(t''(\tau''))
\end{align}
and so $A_-=A_+^\dg,B_-=B_+^\dg$ holds. Note that we set $f(x)$ to be a delta function for convenience. The result of this calculation is oblivious to the details of smearing, window functions or trajectories. The  expression that matters is Eq.~(\ref{eq:twoUDWdetectors}) together with the fact fact that $A_\pm,B_\pm$ are boson scalar fields. Our final goal is to express $2n$-point Green's functions in terms of a polynomial number of two-point correlators and not their calculation where these details are relevant.

The matrix elements of~(\ref{eq:twoUDWdetectors}) will be written as
\begin{equation}\label{eq:omega}
\om_{AB}(\kbr{kl}{mn}) = \bra{0_M} \bra{kl} U_AU_B \ket{ij} \big(\bra{mn} U_AU_B \ket{ij}\big)^\dg \ket{0_M},
\end{equation}
where $i,j,k,l\in\{0,1\}$ denote the canonical basis of the detectors' initial and final state and a tensor product is implied. In all previous examples we calculated a transition amplitude $\bra{f}S^{(m)}\ket{i}$ for a perturbatively evolved $S$-matrix. This is a typical task in QFT. Here we are after a different (but, of course, related) quantity: the transition probability. More precisely, we want to perturbatively expand density matrix~\eqref{eq:omega} where the probabilities lie on the diagonal. We will derive $\om_{AB}$ for $i,j=0$ (that is, assuming the detectors are in their ground states) and later argue that a trivial modification of our analysis will enable us to calculate Eq.~(\ref{eq:omega}) for any canonical basis state $\ket{ij}$. This, on the other hand, will lead to  a `standalone' unitary matrix $U_A\otimes U_B$ expressed perturbatively which then can be used to find $\om_{AB}$ for \emph{any}  detector input state and thus completely solving the problem.

We start by taming the exponential number (in $n$) of the summands of the core expression,~Eq.~(\ref{eq:twoUDWdetectors}).
By plugging the result into~Eq.~\ref{eq:omega} we get again a polynomial number of $2n$-point Green's functions and apply the results obtained in this paper to express them in terms of the two-point Green's functions. Let us recall the algebra of Pauli operators  representing the two-level detector(s). They satisfy $(\s^+)^2=(\s^-)^2=0$ and also $\s_A^\pm\s_B^\pm=\s_B^\pm\s_A^\pm$ together with $\s_A^\pm\s_B^\mp=\s_B^\mp\s_A^\pm$. That is, the atomic operators for the two detectors commute in the Unruh-DeWitt model.
\begin{prop}\label{prop:phaseI}
  For $n=2m$ we obtain from Eq.~(\ref{eq:twoUDWdetectors})
  \begin{align}\label{eq:expEven}
    (A_+\sigma_A^++A_-\sigma_A^-+B_+\sigma_B^++B_-\sigma_B^-)^{2m}\ket{00}
    &=\sum_{\ell=0}^m\binom{2m}{2\ell}\,A_+^\ell A_-^\ell B_+^{m-\ell}B_-^{m-\ell}\ket{00}\nn\\
    &+\sum_{\ell=1}^m\binom{2m}{2\ell-1}\,A_+^{\ell}A_-^{\ell-1} B_+^{m-\ell+1}B_-^{m-\ell}\ket{11}
  \end{align}
  and for $n=2m+1$ we get
  \begin{align}\label{eq:expOdd}
    (A_+\sigma_A^++A_-\sigma_A^-+B_+\sigma_B^++B_-\sigma_B^-)^{2m+1}\ket{00}
    &=\sum_{\ell=0}^m\binom{2m+1}{2\ell}\,A_+^\ell A_-^\ell B_+^{m-\ell+1}B_-^{m-\ell}\ket{01}\nn\\
    &+\sum_{\ell=0}^m\binom{2m+1}{2\ell}\,A_+^{m-\ell+1}A_-^{m-\ell} B_+^{\ell}B_-^{\ell}\ket{10}.
  \end{align}
\end{prop}
\begin{proof}
  The initial state $\ket{0}$ is a ground state: $\s^-\ket{0}=0$. There are two possibilities for the output states: $\ket{0}$ and $\ket{1}$ but taking into account the nilpotency property of $\s^\pm$ we can also enumerate all possible ways the output states can be obtained. It is simply
  \begin{align}\label{eq:allPaths}
    \ket{0} & = (\s^-\s^+)^{p}\ket{0}, \\
    \ket{1} & = (\s^-\s^+)^{p-1}\s^+\ket{0}
  \end{align}
  for all $p>0$. All other operator sequences result in zero. In our case we have two sets of ladder operators -- for the $A$ and $B$ Hilbert spaces and so
\begin{subequations}\label{eq:allPaths2Atoms}
  \begin{align}
    \ket{00}_{AB} & = \sum_{\genfrac{}{}{0pt}{2}{\mathrm{nonzero\ products\ s.t.}}{2p+2q=n}}(\s_A^-\s_A^+)^{p}(\s_B^-\s_B^+)^{q}\ket{00}, \label{eq:allPaths2Atoms00} \\
    \ket{11}_{AB} & = \sum_{\genfrac{}{}{0pt}{2}{\mathrm{nonzero\ products\ s.t.}}{2p+2q-2=n}}(\s_A^+\s_A^-)^{p-1}\s_A^+(\s_B^+\s_B^-)^{q-1}\s_B^+\ket{00}, \label{eq:allPaths2Atoms11} \\
    \ket{01}_{AB} & = \sum_{\genfrac{}{}{0pt}{2}{\mathrm{nonzero\ products\ s.t.}}{2p+2q-1=n}}(\s_A^-\s_A^+)^{p}(\s_B^+\s_B^-)^{q-1}\s_B^+\ket{00}, \label{eq:allPaths2Atoms01} \\
    \ket{10}_{AB} & = \sum_{\genfrac{}{}{0pt}{2}{\mathrm{nonzero\ products\ s.t.}}{2p+2q-1=n}}(\s_A^+\s_A^-)^{p-1}\s_A^+(\s_B^-\s_B^+)^{q}\ket{00}.\label{eq:allPaths2Atoms10}
  \end{align}
\end{subequations}
  In the first two lines $n$ is even and in the last two lines $n$ is odd. The sums should be understood as counting all possible strings of the sigma operators leading to a given output state on the left and we will now find the explicit expressions. Let's call them \emph{nonzero strings}. Given the desired output state,  the main observation is that there are always only two sigma operators that do not destroy the previous nonzero string. This is because every nonzero string acting on $\ket{00}$ can land in only one of the four possible states. For instance, if the state is $\ket{01}$ the next sigma operator can be $\s^+_A$ or $\s_B^-$ to get another nonzero string. Given the offspring node, the choice of either $+$ or $-$ in the branching is unambiguous in order to get a nonzero string. Similarly for the other three states occupying the node and so we can represent all nonzero strings as a complete binary tree of the depth $n$ where we adopt the following convention: The root node is the initial state $\ket{00}$ where the left offspring corresponds to the action of $\s_A^+$ and the right offspring corresponds to $\s_B^+$. The nodes are labeled by the resulting state $\ket{ij}$, in this case $\ket{01}$ or $\ket{10}$, see Fig.~\ref{fig:littlebintree}. In all the branchings that follow the left offsprings correspond to the action of $\s_A^\pm$ and the right offsprings correspond to $\s_B^\pm$.
  \begin{figure}[t]
    \resizebox{4cm}{!}{\includegraphics{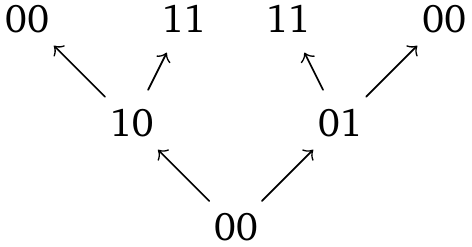}}
    \caption{All paths for the first two branchings starting from $\ket{00}_{AB}$ are depicted. The arrows indicate the action of the ladder operators $\s_{A,B}^\pm$ with $A(B)$ acting to the left(right). The resulting state is a new node.}\label{fig:littlebintree}
  \end{figure}
\begin{figure}[h]
   \resizebox{12cm}{!}{\includegraphics{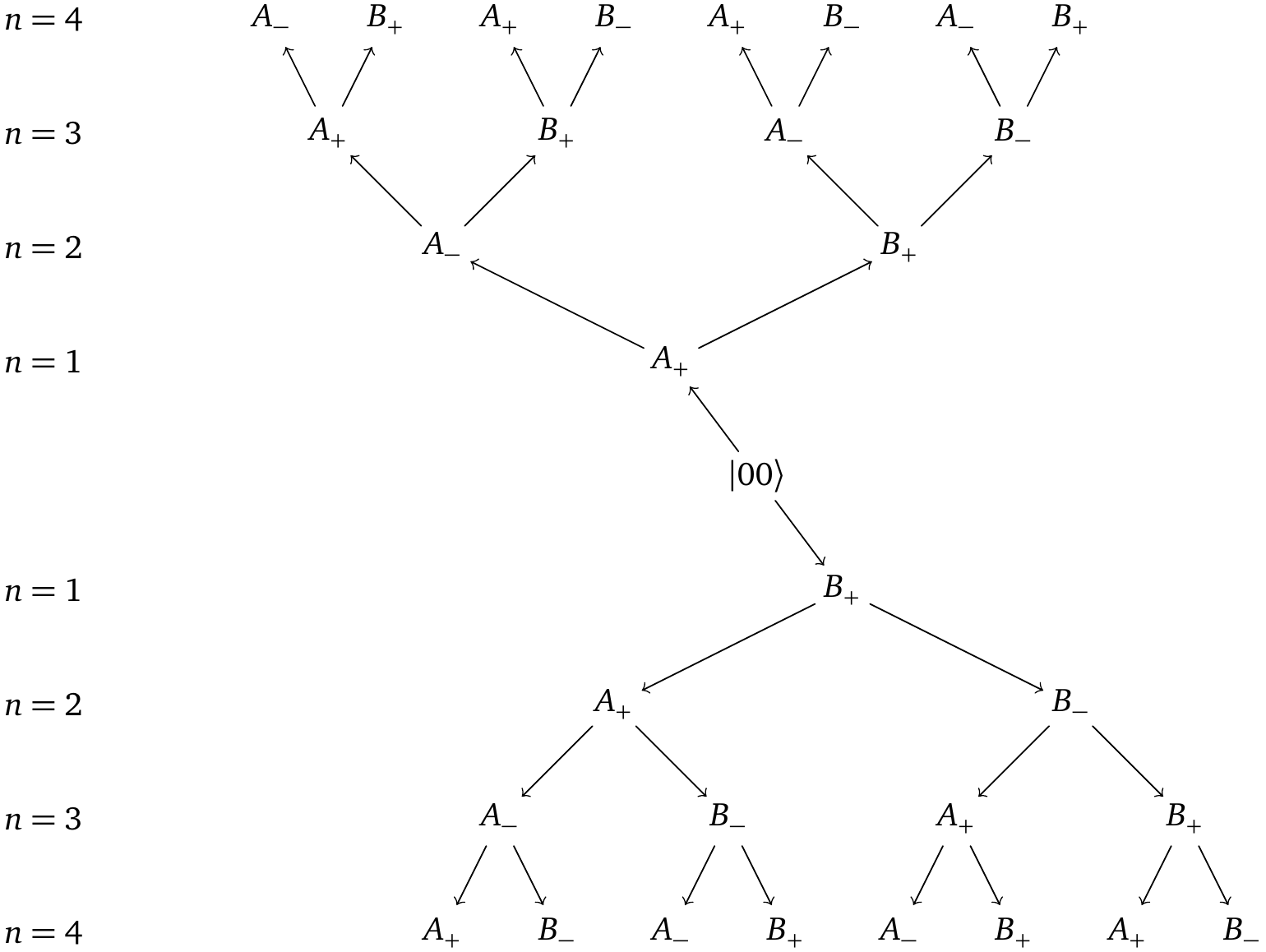}}
    \caption{A schematic depiction of the action of Eqs.~(\ref{eq:expEven}) and~(\ref{eq:expOdd}) for $n=4$ branchings. As in Fig.~\ref{fig:littlebintree} the arrows symbolize the action of the detector ladder operators $\s_{A,B}^\pm$. The left-pointing arrows correspond to the $A$ detector while the right pointing arrows are for the $B$ detector. The initial state is $\ket{00}$ and the symbols $A_\pm$ and $B_\pm$ are the boson operators `left behind' after the action of $\s_{A,B}^\pm$. So as we follow any path in the tree (up or down from $\ket{00}$) we collect the boson operators thus forming an operator sequence in Eqs.~(\ref{eq:expEven}) and~(\ref{eq:expOdd}). There are many paths with the same operator products and their multiplicities are the binomial coefficients derived in Proposition~\ref{prop:phaseI}.}
    \label{fig:bintree}
\end{figure}

  But this immediately provides the desired counting because we have just shown a bijection between the number of branches of a perfect binary tree of the depth $n$ and the number of non-zero strings of the same length. This is because each branch is uniquely represented by a nonzero string of the length $n$ and no string out of $2^n$ of them is missing.  Consider~Eq.~(\ref{eq:allPaths2Atoms00}). We set the total number of branchings $n=2m$ and $p=\ell$ giving us $q=m-\ell$. The number of paths with $2\ell$ left branchings (one $\ell$ for $\s_A^+$ and one $\ell$ for $\s_A^-$ making it even since we have to end up in $\ket{0}_A$)  out of the total $n=2m$ branchings is $\binom{2m}{2\ell}$ and so is the number of nonzero strings in~Eq.~(\ref{eq:allPaths2Atoms00}). Identically, for~(\ref{eq:allPaths2Atoms11}) we again set $n=2m$ and  $p=\ell$ and obtain $\binom{2m}{2\ell-1}$ of nonzero strings ($2\ell-1$ left branchings and $m-\ell+1$ right branchings). For Eq.~(\ref{eq:allPaths2Atoms01}) we set $n=2m+1$ and $p=\ell$. Hence $\binom{2m+1}{2\ell}$ is the total number of nonzero strings. Finally, for Eq.~(\ref{eq:allPaths2Atoms10}) we have $n=2m+1$ and we now set $q=\ell$. In this way we obtain the same number of nonzero strings $\binom{2m+1}{2\ell}$ as in the previous case. We derived the binomial coefficients in Eqs.~\eqref{eq:expEven} and~\eqref{eq:expOdd}. The sums  collect all Pauli operator products leading to the same target state $\ket{ij}$.  The proposition statement is then obtained by inspecting the LHS of Eqs.~(\ref{eq:expEven}) and~(\ref{eq:expOdd}) where the boson operators $A_\pm,B_\pm$ accompany the corresponding ladder operators and the RHS consist of the products of the boson operators from nonzero strings.
\end{proof}
\begin{rem}
    Fig.~\ref{fig:bintree} illustrates the proof. It also illustrates the fact that the same analysis can be done for any initial state $\ket{ij}$. Depending on the values of $i$ and $j$ the arrows in the binary tree will represent different ladder operators leading to the same counting and different boson operators on the RHS of  Eqs.~(\ref{eq:expEven}) and~(\ref{eq:expOdd}). In this way we are able to find all the unitary matrix elements and therefore we can act on any input state (pure or mixed) from a space spanned by $\ket{ij}_{AB}$.
\end{rem}
\begin{exa}
  Setting $m=0,1,2$ in Eq.~(\ref{eq:expEven}) and $m=0,1$ in Eq.~(\ref{eq:expOdd}) (corresponding to the perturbative order $n=4$) we quickly reproduce the fourth order calculation in~\cite{bradler2016absolutely} and the result can be applied to a variety of situations~\cite{ver2009entangling,reznik2005violating,cliche2010information,lin2006accelerated}. With the same ease we can obtain an output state for $m\gg0$.

  A non-trivial check is the trace of~(\ref{eq:omega}) being equal to one irrespective of the perturbative order. Case $n=4$ was verified in~\cite{bradler2016absolutely} and a straightforward calculation shows that it is true for $n=6$ as well\footnote{Note that the fourth order is an absolute must if one wants to calculate any entropic quantity to properly assess the importance of entanglement for quantum communication. If only the second order is calculated (almost always the case with a notable exception of~\cite{bradler2016absolutely}) then the $\kbr{11}{11}$ component of~(\ref{eq:omega}) is zero and two of the eigenvalues of $\om_{AB}$ are negative.}.
\end{exa}

We use Proposition~\ref{prop:phaseI} to insert Eqs.~(\ref{eq:expEven}) and~(\ref{eq:expOdd}) to Eq.~(\ref{eq:omega}) and get the output state components to the $n$-th order. Let's take a look at the $\kbr{00}{00}$ component:
\begin{align}\label{eq:omega00}
  \om_{AB}(\kbr{00}{00})
  &= \bra{0_M}\mathsf{T}\Big\{
    \sum_{m+m'=0}^{n}{\la^{2(m+m')}(-1)^{m+m'}\over (2m)!(2m')!}\bigg(\sum_{\ell=0}^m\binom{2m}{2\ell}\,A_+^\ell A_-^\ell B_+^{m-\ell}B_-^{m-\ell}\bigg)\nn\\
    &\quad\times
    \bigg(\sum_{\ell'=0}^{m'}\binom{2m'}{2\ell'}\,A_+^{\ell'} A_-^{\ell'} B_+^{m'-\ell'}B_-^{m'-\ell'}\bigg)^\dg
  \Big\} \ket{0_M}\nn\\
  &= \bra{0_M}\mathsf{T}{\Big\{
    \sum_{m+m'=0}^{n}{\la^{2(m+m')}(-1)^{m+m'}\over (2m)!(2m')!}
    \binom{2m}{2\ell}\binom{2m'}{2\ell'}
    \,A_+^{\ell+\ell'} A_-^{\ell+\ell'} B_+^{m+m'-\ell-\ell'}B_-^{m+m'-\ell-\ell'}
  \Big\}} \ket{0_M}\nn\\
    &= \sum_{m+m'=0}^{n}{\la^{2(m+m')}(-1)^{m+m'}\over (2m)!(2m')!}
    \binom{2m}{2\ell}\binom{2m'}{2\ell'}
    \lan1^{\ell+\ell'}2^{\ell+\ell'}3^{m+m'-\ell-\ell'}4^{m+m'-\ell-\ell'}\ran_0.
\end{align}
We identified  $A_+\rightleftharpoons1,A_-\rightleftharpoons2,B_+\rightleftharpoons3$ and $B_-\rightleftharpoons4$ and in this form it is ready for Theorem~\ref{thm:IsserlisRefined} and onto Theorem~\ref{thm:Multiplicity}. The number of terms in~Eq.~(\ref{eq:omega00}) is polynomial in $m$ and therefore in $n$ as well.

\subsection*{Hafnians and the number of perfect matchings of a graph}

Following Def.~\ref{def:graph} we introduce some additional concepts from graph theory.
\begin{defi}
  A simple undirected graph $G=(V,E)$ on $m=|V|$ vertices is called complete if all vertices are connected with each other. That is, $|E|=m(m-1)/2$ is the total number of edges. A graph which is not complete is called incomplete.
\end{defi}
\begin{defi}
  A perfect matching (or 1-factor) of a simple undirected graph $G=(V,E)$ is a graph $G_{PM}=(V',E')\subset G$ where every vertex is paired with exactly one other vertex. Hence $G$ contains a perfect matching if $|V|=m$ is even in which case $V'=V$ and $|E'|=|V|/2$.
\end{defi}
The number of perfect matchings of a graph varies, depending on the graph type. The majority of graphs are incomplete and  the calculation of the number of perfect matchings not known to be tractable. One of the very few cases where the analytical form is known is $K_m$ -- a complete graph on $m=2n$ vertices where $\#(PM_{K_{2n}})=(2n)!/(2^nn!)=(2n-1)!!$.
\begin{exa}\label{exa:completeGraphK4}
  The following picture shows all the perfect matching for $K_4$.
    \begin{figure}[h]
    \resizebox{14cm}{!}{\includegraphics{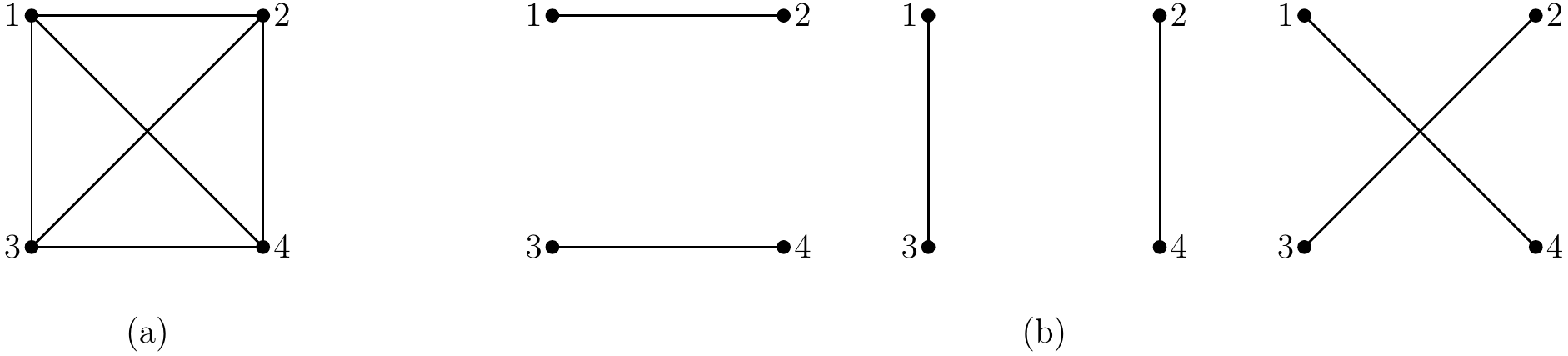}}
    \caption{(a) $K_4$: complete graph on four vertices. (b) All three perfect matchings of $K_4$.}\label{fig:K4}
  \end{figure}

\end{exa}
The picture resembles the way a four-point correlation function splits into as a result of Wick's (Isserlis') theorem. Caianiello formalized this connection~\cite{caianiello1953quantum} by introducing the hafnian of a $2n\times2n$ matrix~$A$ (here we use the definition from~\cite{minc1984permanents})
\begin{equation}\label{eq:haf}
  \haf{[A(K_4)]}=\sum_{\s}\prod_{i=1}^{2n}A(K_4)_{i,\s(i)},
\end{equation}
where the sum goes over all $(2n-1)!!$ unordered disjoint pairs of the set $\{1,\dots,2n\}$. For the above example of a complete graph on four vertices we calculate the hafnian of its adjacency matrix
$$
A(K_4)=\begin{bmatrix}
    0 & 1 & 1 & 1 \\
    1 & 0 & 1 & 1 \\
    1 & 1 & 0 & 1 \\
    1 & 1 & 1 & 0 \\
  \end{bmatrix}
$$
and find $\haf{[A(K_4)]}=1+1+1=3$ agreeing with Example~\ref{exa:completeGraphK4}. Indeed, the hafnian of any complete graph is precisely the number of summands of the products of two-point correlation functions in Wick's (Isserlis') formula, Eq.~\ref{eq:Isserlis}, as envisaged by Caianiello.

Our contribution in this section is the following. While the number of perfect matchings for a complete graph is known we may face the problem of systematically listing all of them. Following the proof of Theorem~\ref{thm:IsserlisRefined}, namely the set of Diophantine equations in~\eqref{eq:DiophantGeneral}, we set $f=2n$ and $\ell_i=1,\forall i$ and by solving the system obtain the desired perfect matchings.

\section{Generalization beyond boson scalar theories}\label{sec:beyondscalars}

Wick's theorem is at the core of any perturbative calculation and so the method based on finding all nonnegative solutions of a Diophantine system can in principle be used everywhere. But there are several, mainly technical, obstacles in extending the current formalism to more general theories. One of the was mentioned in Section~\ref{sec:bosonsInteracting} regarding the presence of derivative interactions in the Lagrangian.

Another obstacle to overcome in order to generalize the developed scheme is to consider Wick's theorem for anticommuting fields. This is a standard tool in QFT~\cite{schwartz2014quantum} where one has to keep track of the parity and change the sign whenever two fermions swap their place in a correlation function. There seems to be no equivalent of Isserlis' theorem for anticommuting variables but one can envisage a generalization of random Gaussian variables to a Gaussian integral over (anticommuting) random Grassmann variables. After all, Grassmann variables are routinely used to study the properties of fermions in the path integral formulation of QFT using \emph{similar} algebraic properties of the CAR and Grassmann algebra. Then, a generalization of our refinement of Wick's/Isserlis' theorem (Theorem~\ref{thm:IsserlisRefined}) is straightforward as well but due to the nilpotency of Grassmann variables there is really not much to generalize as the higher powers are missing.

Once we know how treat anticommuting fields we proceed in the same way by solving a system of Diophantine equations. The difference is (and this is true for any complex boson theory as well) that we have to disregard the solutions where a contraction of two fields is zero. For example, for the second order expansion of the Yukawa interaction term $\Lcal_{\mathrm{int}}=-g\psi^\dg(x)\psi(x)\phi(x)$, where $\psi$ is anticommuting, the contribution given by contracting $\psi(x)\psi(y)$ is zero. The Diophantine system, of course, treats all fields equally and so we have to exclude such contributions manually. But this can be easily done just by going through the list of Diophantine solutions and removing all those containing at least one zero propagator such as $\bra{0_M}\mathsf{T}\{\psi^\dg(x)\psi^\dg(y)\}\ket{0_M},\bra{0_M}\mathsf{T}\{\psi(x)\phi(x)\}\ket{0_M}$ or any other kinematically forbidden process. By easily we mean without an additional computational complexity cost.

Finally, one has to face the calculation of the graph automorphism order to get the right multiplicities. Already starting with complex boson theories the graphs are in general directed, distinguishing between scattering of particles and antiparticles. This will result in a rather minor modification of Lemma~\ref{lem:FeynmanMulti}. In particular, directed edges or edges of different types (different fields) will decrease the grafted graph symmetry compared to the real scalar case.

\section{Discussion and open questions}\label{sec:discussions}

Among several open problems stands out the following one. We have seen that in order to get the right multiplicity factor of a perturbative contribution one has to use a potentially computationally expensive procedure of calculating the graph automorphism order or graph isomorphism membership problem. This procedure bundles the graphs from the same class of physically indistinguishable scattering processes. Recall that they are represented by different solutions of a Diophantine system and every such a solution is an interior point of a certain convex polyhedron. Could it be that solutions from the same class form an exceptional subset in the polyhedron (perhaps even with nice geometric properties) or are they scattered randomly? In the former case we could use such knowledge to avoid the potentially costly  calculation of the graph isomorphisms (costly for large graphs in extremely high perturbative orders).

A closely related question is whether the study of the contributions from one perturbative order (for a given Lagrangian and a fixed number of interacting fields) tells us something about the structure of the contributions from higher perturbative orders. All nonnegative solutions come from a Diophantine system whose  size  (the number of interior lattice points) grows but its structure is very similar across the perturbative orders. Does geometry of the corresponding convex polyhedra have anything in common across perturbative orders? If the answer is affirmative, could a new insight be gained into how to sum the related perturbative contributions?

Another outstanding problem is a generalization of the major application of our method: an efficient perturbative calculation of the interaction Hamiltonian for a pair of Unruh-DeWitt (UDW) detectors in Minkowski spacetime. The result is independent on the detector details and works for any smearing, envelope function or trajectory by efficiently  decomposing a generic $2n$-point correlation function to a product of two-point Green's functions where the detector details play a role. An efficient calculation means that the number of perturbative terms increases polynomially with the perturbative order of the coupling constant. It would be interesting to generalize the proofs for any number of Unruh-DeWitt detectors to study the multipartite structure of Minkowski and other spacetimes.

A related open problem is to clarify whether the UDW perturbative expansion suffers from the convergence issues. Everything suggests that the series is not of an asymptotic character. It is unlikely but not entirely impossible that the density matrix components or even the unitary entries themselves are actually analytically summable. To answer these questions one could perhaps first look at single UDW detector coupled to Minkowski vacuum as the simplest case. Another interesting question is what happens with the divergencies in the momentum space after all perturbative components are added. Is renormalization still necessary as claimed in~\cite{hummer2016renormalized} or do the divergencies cancel out?

\section*{Acknowledgements}
This material is based upon work supported by the Air Force Office of Scientific Research under award number FA9550-17-1-0083.

\bibliographystyle{unsrt}


\end{document}